\documentclass[letterpaper,12pt]{article} 
\usepackage{amsmath,amsthm,amssymb,amsfonts}
\usepackage{graphicx,float}
\usepackage{osajnl2}

\theoremstyle{remark}               

\newtheorem*{theorem}{\bf Theorem}
\newcommand{\e}{\mathrm{e}} 
\renewcommand{\i}{\mathrm{i}} 
\renewcommand{\d}{\,\!\operatorname{d}\!}
\newcommand{\abs}[1]{\left| {}_{}^{} {#1}{}_{}^{} \right|}
\newcommand{\norm}[1]{\left\|{#1}\right\|}
\renewcommand{\vec}[1]{\boldsymbol{#1}}
\newcommand{\unitvec}[1]{\boldsymbol{\widehat{#1}}} 
\newcommand{\bigO}{\mathcal{O}}
\newcommand{\erf}{\operatorname{erf}} 

\renewcommand{\eqref}[1]{Eq.(\ref{#1})}

\begin{document}

\title{Compressive Imaging of Subwavelength Structures II. Periodic Rough Surfaces }
\author{Albert C. Fannjiang$^{*}$ and Hsiao-Chieh Tseng}
\address{Department of Mathematics, University of California, Davis, \\
One Shields Ave., University of California, 
Davis, CA 95616-8633, USA}
\address{$^*$Corresponding author: fannjiang@math.ucdavis.edu}


\date{\today}

\begin{abstract}
A compressed sensing scheme for near-field imaging of corrugations of relative sparse Fourier components is proposed.
The scheme employs  random sparse measurement  of near field
to recover the angular spectrum of the scattered field. 
It is shown heuristically and numerically that under the Rayleigh hypothesis
the angular spectrum is compressible and
amenable to compressed sensing techniques.  

Iteration schemes are developed for recovering the  surface profile from the angular spectrum.  
 The proposed nonlinear least squares in the Fourier basis produces
accurate reconstructions even when the Rayleigh hypothesis
is known to be false.

\end{abstract}

\maketitle

  \section{Introduction}
    Rough surface scattering is of fundamental interest in optics,
  radiowave propagation
  and acoustics \cite{Bec, BF, U65} and forms the basis of
  near-field imaging which is the operation principle
  behind such instruments as scanning near-field optical microscopy \cite{AN, Poh, Lew, PDL} and near field acoustic microscopy \cite{KY91}.  Near-field imaging is a microscopic technique that breaks
  the  diffraction limit by exploiting the properties of evanescent waves.
  The signal  is collected  by placing the detector in a
  distance much smaller than wavelength $\lambda$ to the specimen surface. An image of the surface is obtained by mechanically moving the probe in a raster scan of the specimen, line by line, and recording the probe-surface interaction as a function of position.  This leads
  to  long scan times for large sample areas or high resolution imaging.

  Typically near-field imaging is analyzed by assuming a continuum or
  dense set of data points \cite{GV93,  NG81, RGC81}.
  In the present work, we focus on the setting of  {\em sparse, discrete }  measurement of near-field
  from the perspective of  compressed sensing theory. This is an extension of the work \cite{subwave} on potential scattering to the case of rough surface scattering. Surface scattering involves the geometry (i.e. topography)  of scatterers and is technically more challenging to deal with  than
  potential scattering.

 Consider the scattering problem
  for a corrugation profile described by the function $z=h(x)$. For simplicity of presentation, we will focus on the case of two-dimensional  scalar wave with the Dirichlet boundary condition.
  The total field $u^{\text{tot}}$ satisfies
   \begin{eqnarray}
   \Delta u^{\text{tot}} + k^2 u^{\text{tot}} \ = \ 0  &\qquad \text{ in } \Omega\subset\mathbb{R}^2 \ , \quad k>0 \\
   u^{\text{tot}} \ = \ 0 &\qquad \text{ on } \partial \Omega,
  \end{eqnarray}
  where
  \begin{eqnarray}
  \Omega = \{\vec{r}=(x,z)\in\mathbb{R}^2\colon z>h(x)\}, \quad   h \in C(\mathbb{R})\cap L^\infty(\mathbb{R}) .
  \end{eqnarray}

    The total field models the sound pressure wave or
  electromagnetic waves in the TE-mode. The Dirichlet boundary condition
corresponds to  the sound-soft boundary condition in acoustics  and in electromagnetism the perfectly conducting boundary condition.
Our approach can be easily
  extended to the three dimensional case as well as to
 the Neumann boundary condition, corresponding to acoustically hard obstacles, and the Robin boundary condition.


  As usual in scattering problem, we write   $u^{\text{tot}}=u^{\text{inc}}+u$
  where both the scattered wave $u$ and the incident
  wave $u^{\text{inc}}$
  satisfy the Helmholtz equation. The Dirichlet  condition becomes $u=-u^{\text{inc}}$ on $\partial \Omega$.

  In this  paper, we focus on the case of  periodic surfaces which 
  include diffraction gratings, an important class ofoptical elements. 
We  assume that $h$ has period $L$
  and $u^{\text{inc}}$ is the plane incident wave
  \begin{eqnarray}u^{\text{inc}}(\vec{r}) = \e^{\i k \unitvec{\mathbf d}\cdot \vec{r}}
  =\e^{\i k(x\cos\theta-z\sin\theta)}\ , \quad
  \unitvec{\mathbf d} = (\cos\theta,-\sin\theta)
  \ , \quad 0 < \theta < \pi .
  \end{eqnarray}

Observe that on the boundary $z=h(x)$
  \begin{eqnarray}u(x+L,h(x)) = -\e^{\i k \big((x+L)\cos\theta-h(x)\sin\theta\big)}
  = \e^{\i k L\cos\theta} u(x,h(x)).\end{eqnarray}
Hence we look for  the  $(L,k\cos\theta)$-quasi-periodic (or Floquet periodic) solution satisfying
  \begin{eqnarray}u(x+nL,z) = \e^{\i nL k \cos\theta} u(x,z)
  \ \text{ for all } \ (x,z)\in\Omega ,\ n\in\mathbb{Z}.\end{eqnarray}
In particular,  if $\theta=\pi/2$, then $u$ is $L$-periodic. To fix the idea, we set $L=2\pi$.

  \begin{figure}[h]\centering

\ifx\JPicScale\undefined\def\JPicScale{0.8}\fi
\unitlength \JPicScale mm
\begin{picture}(120,80)(0,0)
\linethickness{0.3mm}
\qbezier(10,10)(21.25,13.75)(28.75,13.75)
\qbezier(28.75,13.75)(36.25,13.75)(40,10)
\qbezier(40,10)(43.75,6.25)(47.5,6.25)
\qbezier(47.5,6.25)(51.25,6.25)(55,10)
\qbezier(55,10)(58.44,13.59)(60,14.38)
\qbezier(60,14.38)(61.56,15.16)(65,15)
\qbezier(65,15)(68.75,14.38)(72.5,10)
\qbezier(72.5,10)(76.25,5.62)(80,5)
\qbezier(80,5)(82.81,4.84)(83.75,5.62)
\qbezier(83.75,5.62)(84.69,6.41)(95,10)
\qbezier(95,10)(106.41,13.53)(113.12,13.88)
\qbezier(113.12,13.88)(119.84,14.22)(120,14)
\linethickness{0.3mm}
\multiput(40,70)(0.12,-0.18){167}{\line(0,-1){0.18}}
\put(60,40){\vector(2,-3){0.12}}
\linethickness{0.3mm}
\multiput(50,40)(0.18,0.12){83}{\line(1,0){0.18}}
\linethickness{0.3mm}
\multiput(48,42)(0.18,0.12){83}{\line(1,0){0.18}}
\linethickness{0.3mm}
\multiput(46,44)(0.18,0.12){83}{\line(1,0){0.18}}
\linethickness{0.3mm}
\multiput(45,47)(0.18,0.12){83}{\line(1,0){0.18}}
\linethickness{0.3mm}
\multiput(43,49)(0.18,0.12){83}{\line(1,0){0.18}}
\linethickness{0.1mm}
\multiput(35,65)(1.94,0){16}{\line(1,0){0.97}}
\linethickness{0.1mm}
\qbezier(52,65)(51.75,60.69)(50,59.75)
\qbezier(50,59.75)(48.25,58.81)(48,59)
\put(48,60){\vector(0,0){0.12}}
\linethickness{0.3mm}
\put(15,10){\line(1,0){100}}
\put(115,10){\vector(1,0){0.12}}
\linethickness{0.3mm}
\put(20,5){\line(0,1){70}}
\put(20,75){\vector(0,1){0.12}}
\linethickness{0.3mm}
\put(110,8){\line(0,1){4}}
\put(108,3){$L$}

\put(21,5){$0$}

\put(116,9){$x$}

\put(19,76){$z$}

\put(48,61){$\theta$}

\put(60,18){$h$}

\put(86,56){$\Omega$}

\end{picture} \\
   \caption{Surface topography.}
  \end{figure}
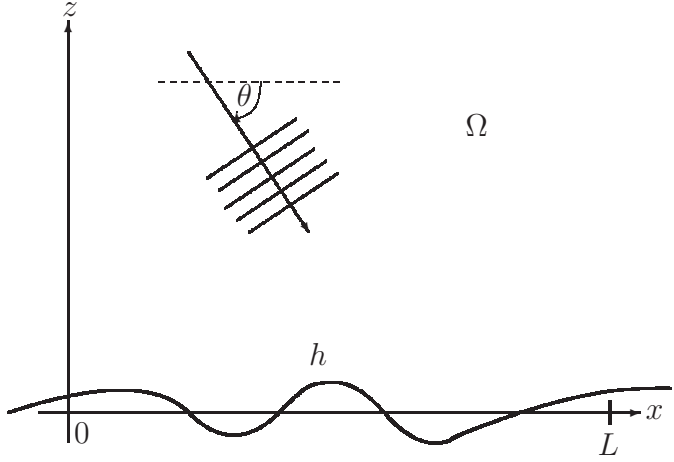
  \section{Radiation condition and Rayleigh hypothesis}
    The existence and uniqueness can be proved under the quasi-periodicity
  and the radiation conditions  on the solution $u$ \cite{K93}.
The well-posedness for  general nonperiodic  rough surfaces is given in  \cite{AH05}. Below we discuss the Fourier representation of
the scattered field and the associated Rayleigh hypothesis.

 For $x\in[-\pi,\pi)$,
  $z > \sup{h} \triangleq h_{\max}$, we write  the scattered field
 as the Fourier series
  \begin{eqnarray}
  u(x,z) = \sum_{n\in\mathbb{Z}} u_n(z)\e^{\i(n+k\cos\theta)x} = \sum_{n\in\mathbb{Z}}u_n(z)\e^{\i k\alpha_n x}
  \end{eqnarray}
  with
   \begin{eqnarray}
  \label{1.0}
  \alpha_n &= &\frac{n}{k}+\cos\theta
  \end{eqnarray}
  where
 $u_n$ satisfies
  \begin{eqnarray}
  \label{1.1}
  \ddot{u}_n + k^2(1-\alpha_n^2)u_n = 0.
  \end{eqnarray}
Solving \eqref{1.1} and imposing the boundedness of $u_n$ at $z=\infty$   we obtain the general solution as
\begin{eqnarray}
    u(x,z) =
    & \sum_{\abs{\alpha_n}\leq 1}a_n\e^{\i k(\alpha_n x - \beta_n z)}   & \text{ (incoming waves) }  \nonumber \\
    & + \sum_{\abs{\alpha_n}\leq 1}b_n\e^{\i k(\alpha_n x + \beta_n z)} & \text{ (outgoing waves) }  \nonumber \\
    & + \sum_{\abs{\alpha_n} > 1}c_n\e^{\i k(\alpha_n x + \beta_n z)}   & \text{ (evanescent waves) } \label{1.2}
\end{eqnarray}
where  $\beta_n$ is given by
\begin{eqnarray}
 \label{1.3}
 \beta_n = \begin{cases}
  \sqrt{1-\alpha_n^2},  & \abs{\alpha_n}\leq 1 \\ \i\sqrt{\alpha_n^2-1},  & \abs{\alpha_n}>1.  \end{cases}\label{defab}
  \end{eqnarray}

  The
   \emph{Rayleigh radiation condition} for the region above
   the grooves  $z>h_{\max}$ amounts to
   dropping  the incoming waves in \eqref{1.2}:
  \begin{eqnarray}
  \label{1.4}
  u(x,z) = \sum_{n\in\mathbb{Z}} u_n\e^{\i k(\alpha_n x+\beta_n z)}
  \ ,\ z> h_{\max}.
  \end{eqnarray}
 However, in the region inside the grooves
 $z<h_{\max}$  multiple scattering may occur
  and \eqref{1.4} may not  represent the  true
  scattered  wave
  in this region.   For shallow  corrugations, \eqref{1.4}
  should hold in the grooves and this is the Rayleigh hypothesis.
  For instance, the  Rayleigh hypothesis holds  for the sinusoidal  profile
$h(x) = b\sin(ax)$ with  $\abs{ab}<0.448$ \cite{M69,M71,U65}.
  On the other hand, for a  general periodic surface the validity
  of the Rayleigh hypothesis may be difficult to assess \cite{Kel}.
  The failure of the Rayleigh hypothesis \eqref{1.4}  manifests in
  the breaking down of the analytic continuation of \eqref{1.4}
  inside grooves.


%

  \section{Inverse scattering formulation}
  Inverse scattering seeks to reconstruct $h(x)$ by transmitting
  incident wave $u^{\text{inc}}$ and measuring the
  scattered field $u$ at certain locations. Moreover, in order to resolve subwavelength structure which is hidden in
  the evanescent waves the measurement should be carried out in the near-field.

  Due to the quasi-periodicity, we may consider the scattered field $u$
  in  the union of
  \begin{eqnarray} \Omega_\Gamma\triangleq\{(x,z)\in \Omega,\ x\in [-\pi,\pi)\}  \end{eqnarray}
  and
  \begin{eqnarray} \Gamma\triangleq \{(x,z)\in\partial \Omega\colon x\in[-\pi,\pi)\}.
  \end{eqnarray}
  For $z>h(x)$ we have the outgoing  scattered wave representation \eqref{1.4} with \cite{ACD06}
  \begin{eqnarray}
   u_n = \frac{\i}{4\pi k\beta_n }\int_{-\pi}^{\pi} \e^{-\i k\big(\alpha_n x' + \beta_nh(x')\big)}
   \left(-\frac{\partial  u^{\text{tot}}(\vec{r}')}{\partial \nu'}\Big|_{\vec{r}'\in\Gamma} \right)
   \sqrt{1+\dot{h}^2(x')}\d x' .\label{1.15}
  \end{eqnarray}
  To ensure $\beta_n\neq 0$ in \eqref{1.15}, we assume  that $\abs{\alpha_n}\neq 1$,
  i.e.,
  \begin{eqnarray}
  \abs{\cos\theta+\frac{n}{k}}\neq 1,\quad n\in\mathbb{Z}
    \label{1.16}
  \end{eqnarray}
  to avoid all grazing modes. In the case of normal incidence
  $\theta=\pi/2$, \eqref{1.16} means that the wavenumber $k$ is
  not an integer.

A key assumption for our approach is that $h$ has a {\em small number}  of
significant  Fourier coefficients, namely the Fourier coefficients are
sparse or compressible.
Writing $h(x) = \sum_{n\in\mathbb{Z}} \widehat{h}_n\e^{\i n x}$, we say that $\widehat{\vec{h}}=\{\widehat{h}_n\}$
  is $s$-sparse if $\|\widehat{\vec{h}}\|_0$, the number of nonzero elements of $\widehat{\vec{h}}$, is less
  or equal to a small integer $s$. Note that ${\widehat{h}^*}_{-n} = \widehat{h}_n$ since $h$ is real-valued.
  Without loss of generality, we assume that $\widehat{h}_0=0$.

For reconstruction of  $h(x)$ we  utilize the sparsity of $\widehat{\vec{h}}$ which surprisingly  yields
  compressibility of the scattering amplitude $u_n$. Compressive sensing techniques can then used to effectively
  recover those modes.

Let $(x_j, z_0) $, $j=1,2,\ldots,m$, be the sensor  locations  for
 measuring the scattered field where   $z_0> h_{\max}$ is fixed and
  $x_j$
are randomly and  independently
 chosen from $[-\pi,\pi)$ according to the uniform distribution.

  In view of the identify
  \begin{eqnarray}u(x_j,z_0) \e^{-\i k\cos\theta x_j} =
  \sum_{n\in\mathbb{Z}}\e^{\i n x_j} u_n\e^{\i k \beta_n z_0}\end{eqnarray}
  from \eqref{1.4},
let us  consider the following inverse problem
 $Y = {\mathbf{A}} X$, with entries
  \begin{eqnarray}  X_n &= &u_n \e^{\i k \beta_n z_0}\sqrt{m} \label{1.7}\\
Y_j &=& u(x_j,z_0)\e^{-\i k x_j \cos\theta}\label{1.8}\\
  {\mathbf{A}} &=& [A_{j,n}] = \frac{1}{\sqrt{m}}\e^{\i n x_j } \label{CSsystem}\label{1.9}
  \end{eqnarray}
  where $n$ is restricted to a finite, but sufficiently large
  interval
  ranged from $-N/2$ to $N/2-1$.
  In general the system \eqref{1.7}-\eqref{1.9} is highly
  underdetermined for any $m<\infty$.

Surprisingly,  sparse Fourier coefficients $\{\widehat h_n\}$ give rise to  sparse or compressible $\{u_n\}$ and therefore $X$
which can be reconstructed by   compressed sensing.

  \section{Compressive sensing (CS)}
The main thrust of compressed sensing  \cite{CT, Don} is to convert the noisy underdetermined system
\begin{eqnarray}
\label{1.17}
Y={\mathbf{A}} X+E
\end{eqnarray}
 into the $L^1$-based optimization problem
\begin{eqnarray}
   \label{1.10}
   \min\norm{X}_1 \ \text{ subject to } \norm{Y-{\mathbf{A}} X}_2\leq\epsilon \triangleq \|E\|_2
\end{eqnarray}
where $E$ is the external noise vector. \eqref{1.10} is called the Basis Pursuit (BP) \cite{CDS}. In addition to quadratic programming, 
many iterative and greedy  algorithms are available for solving the system \eqref{1.17}.

 Let us first review a basic notion in CS which provides a performance guarantee for BP.
  We say a matrix ${\mathbf{A}}\in\mathbb{C}^{m\times N}$ satisfies the restricted isometry property  (RIP)
  if
  \begin{eqnarray}
  \label{1.11} (1-\delta)\norm{Z}_2^2 \leq 
  \norm{{\mathbf{A}} Z}_2^2\leq (1+\delta)\norm{Z}_2^2,\quad \delta\in (0,1)
  \end{eqnarray}
  holds for all $s$-sparse $Z \in\mathbb{C}^{N}$.
  The smallest constant  satisfying \eqref{1.11} is called
  the restricted isometry constant (RIC) of order $s$  and
  denoted by $\delta_s$.

  The following theorem says that the random Fourier matrix satisfies RIP if $m$ is sufficiently large.
  \begin{theorem}
   \cite{R08}  Let $\xi_j\in[0,1]$, $j=1,2,\ldots,m$ be independent uniform random variables. If
   \begin{eqnarray}\frac{m}{\ln m} \geq C\delta^{-2}s\ln^2 s\ln N\ln\frac{1}{\eta}\ , \quad \eta\in(0,1)\end{eqnarray}
   for some universal constant $C$ and sparsity level $s$, then the restricted isometry constant of the random
   Fourier measurement matrix with
   \begin{eqnarray}
   A_{nj} =\frac{1}{\sqrt{m}}\e^{2\pi\i n \xi_j}, n=-N/2,\ldots,N/2-1,
   \label{rfm}
   \end{eqnarray}
    satisfies
   $\delta_s\leq \delta$ with probability at least $1-\eta$.
  \end{theorem}

 Denote $X_s$ to be the
  best $s$-term approximation of the solution $X$, and let $\tilde{X}$ be the solution of BP \eqref{1.10}.
  \begin{theorem}
   \cite{C08} Let ${\mathbf{A}}$ satisfy the RIP with
   \begin{eqnarray}
   \delta_{2s}<\sqrt{2}-1
   \end{eqnarray}
    and $\tilde X$ be the solution to BP.
   Then
   \begin{eqnarray}\norm{\tilde{X}-X}_2\leq C_0 \frac{1}{\sqrt{s}}\norm{X_s-X}_1 + C_1\epsilon
   \ , \quad \norm{\tilde{X}-X}_1\leq C_0 \norm{X_s-X}_1 + C_1\epsilon\end{eqnarray}
   for some constants $C_0,C_1$ independent of $X$.
  \end{theorem}

Once the estimate $\tilde X$ is obtained from BP, we
reconstruct $u_n$ by
  \begin{eqnarray}
  \tilde{u_n} = \frac{1}{\sqrt{m}} \e^{-\i k \beta_n z_0} \tilde X_n \label{Xtoun}.
  \end{eqnarray}
 The problem with \eqref{Xtoun} is
 that the evanescent modes  yield exponentially
 large factor $\e^{-\i k\beta_n z_0}$
for
  \begin{eqnarray}
  \label{1.12}
  \abs{\alpha_n}=\abs{\cos\theta+\frac{n}{k}}>1.
  \end{eqnarray}
 For $n$ sufficiently large, this can magnify the error
 in $\tilde X_n$ and produce undesirable result in
 $\tilde{u_n}$. This observation also shows
 that $X$ may be much more compressible than
 $\{u_n\}$.

 A simple remedy would be to apply the hard thresholding by restricting the identity \eqref{Xtoun} up to
 $n_0$ sufficiently small and setting  the rest of $\tilde{u_n}$
 zero for $\abs{n}>n_0$. Let us now give a rough estimate for the number
 of modes that should be preserved  by the hard   thresholding rule.

We define
  the stably recoverable evanescent modes
  to be those modes satisfying  \eqref{1.12} and
  \begin{eqnarray}k\abs{\beta_n}z_0\leq C_e  \end{eqnarray}
  for some constant $C_e$ (in \cite{subwave}, $C_e=2\pi$).
 On the other hand,
  \begin{eqnarray}
   \beta_n = \sqrt{\cos^2\theta+2{\frac{n}{k}} \cos\theta+\frac{n^2}{k^2}-1} \geq\frac{\abs{n}}{k}-1 .\end{eqnarray}
  Hence  the stably recoverable modes necessarily satisfy
  \begin{eqnarray}k \left(\frac{\abs{n}}{k}-1\right)  z_0\leq k\abs{\beta_n}z_0\leq C_e \end{eqnarray} or equivalently
  \begin{eqnarray}
   \abs{n}\leq n_0\triangleq \frac{C_e}{z_0}+k
  \end{eqnarray}
 which is a rough characterization of the stably recoverable (evanescent)
 modes.   We see that  $n_0$ increases as $k$ increases or $z_0$ small.

  Summing up the previous analysis we conclude the recoverability of the scattering
 amplitude $\{u_n\}$ by the following theorem:
  \begin{theorem}
   Let $z_0>h_{\max}$ be fixed and let $x_j$, $j=1,2,\ldots,m$ be i.i.d  uniform random variables   in $[-\pi,\pi)$.
   Let $n_0=\frac{C_e }{z_0}+k$ for some positive constant $C_e>0$.
  Let $\tilde X$, $X_s$ be the BP solution and the best $s$-term approximated solution of the system \eqref{1.17}
  respectively, and assume
   \begin{eqnarray}\frac{m}{\ln m} \geq C\frac{1}{\delta^2}s\ln^2 s \ln N \ln\frac{1}{\eta}\end{eqnarray}
   for some universal constant $C$ and $\eta\in (0,1)$.
   Let $\vec{u}=(u_{-n_0},\ldots,u_{n_0})$,
   $\tilde{\vec{u}}=(\tilde{u_{-n_0}},\ldots,\tilde{u_{n_0}})$
   where $\tilde{u_n}$  is given by \eqref{Xtoun}.
   Then one can reconstruct
   the solution $\tilde{\vec{u}}$ with
   \begin{eqnarray}
    \norm{\tilde{\vec{u}}- \vec{u} }_2
    \leq \frac{\e^{C_e}}{\sqrt{m}}\left( C_0 \frac{1}{\sqrt{s}}\norm{X_s-X}_1 + C_1\epsilon  \right)
    \label{thm}
   \end{eqnarray}
   for some constants $C_0,C_1$
   with probability at least $1-\eta$.
  \end{theorem}
\begin{proof}
  Without loss of generality,
  we prove for the case that $x_j$, $j=1,2,\ldots,m$, are i.i.d  uniform random variables in $[0,2\pi)$,
  and consequently the matrix $A$, defined in \eqref{1.17}, is the random Fourier measurement \eqref{rfm} where $\xi_j$ are i.i.d
  uniform random variables in $[0,1)$. It is equivalent to the case where
  $x_j$ are i.i.d  uniform random variables in $[-\pi,\pi)$: One
  can write $x_j=2\pi (\xi_j - \frac{1}{2})$ and the sensing matrix
  $A$ is then
  \begin{eqnarray}[A_{j,n}] = \frac{1}{\sqrt{m}}\e^{\i n 2\pi (\xi_j - \frac{1}{2})} =
  \frac{1}{\sqrt{m}}\e^{2\pi \i n\xi_j} (-1)^n .\end{eqnarray}
  By combining the factor $(-1)^n$ into $X_n$ and writing $W_n=(-1)^n X_n$,
  $n=-N/2,\ldots,N/2-1$, we have
  $\norm{\tilde{W}-W}_2=\norm{\tilde{X}-X}_2$, $\norm{W_s-W}_1=\norm{X_s-X}_1$, where
  $\tilde{W}$ and $W_s$ are defined in the same manner.

  Under the assumption of the matrix $A$, we have the estimate
  \begin{eqnarray}\norm{\tilde{X}-X}_2\leq C_0 \frac{1}{\sqrt{s}}\norm{X_s-X}_1 + C_1\epsilon\end{eqnarray}
  for a desired sparsity level $s$ for some constants $C_0,C_1$
   with probability at least $1-\eta$. On the other hand,
  \begin{eqnarray}
   \norm{\tilde{X}-X}_2^2 = \sum_{n=-n_0}^{n_0}
    \abs{\tilde{u_n}\e^{\i k \beta_n z_0}\sqrt{m} - u_n \e^{\i k \beta_n z_0}\sqrt{m}}^2
    +
    \sum_{n\in\Lambda}
    \abs{ \tilde{X}_n - X_n}^2
   \end{eqnarray}
   where $\Lambda\triangleq\{-N/2,\ldots,-n_0+1,n_0+1,N/2-1\}$.
   Moreover, for $\abs{n}\leq n_0$ we have $0<\abs{\e^{-C_e}}\leq \abs{\e^{\i k \beta_n z_0}} \leq 1$, which gives
   \begin{eqnarray}
   \norm{\tilde{X}-X}_2^2 \geq m\abs{\e^{-C_e}}^2\sum_{n=-n_0}^{n_0}
    \abs{\tilde{u_n} - u_n }^2 + 0 = \frac{m}{\e^{2C_e}} \norm{\tilde{\vec{u}} - \vec{u} }_2 ^2
   \end{eqnarray}
   where $\vec{u}=(u_{-n_0},\ldots,u_{n_0})$ and
   $\tilde{\vec{u}}=(\tilde{u_{-n_0}},\ldots,\tilde{u_{n_0}})$.
   Combining these inequalities, we have the estimate that one can reconstruct $\tilde{\vec{u}}$
   with
   \begin{eqnarray}
    \norm{\tilde{\vec{u}}- \vec{u} }_2
    \leq \frac{\e^{C_e}}{\sqrt{m}}\left( C_0 \frac{1}{\sqrt{s}}\norm{X_s-X}_1 + C_1\epsilon  \right).
   \end{eqnarray}
   for some constants $C_0,C_1$
   with probability at least $1-\eta$.

\end{proof}

  \section{Compressibility of the angular spectrum}
  Let us now analyze the compressibility of coefficients $\{u_n\}$.
  We present a heuristic  argument suggesting that 
  the angular spectrum of the scattered field is sparse
  for shallow corrugations. 


Assuming  the validity of  the Rayleigh Hypothesis we have  \begin{eqnarray}
  -u^{\text{inc}}(x,h(x)) = u(x, h(x))=
  \sum_{n\in\mathbb{Z}} u_n \e^{\i k (\alpha_n x + \beta_n h(x))} , \label{etbc}
  \end{eqnarray}
  or equivalently
  \begin{eqnarray}-\e^{-\i k h(x)\sin\theta} = \sum_{n\in\mathbb{Z}}u_n \e^{\i k \beta_n h(x)}\e^{\i n z} .\end{eqnarray}
  For sufficiently flat and smooth surface $h$
  the nearly normal incidence  $\theta\approx\frac{\pi}{2}$ tends to produce nearly  specular diffracted wave \cite{RGC81}
  and hence
  $\{u_n\}$ is
  concentrated at $n=0$. This observation suggests that
  it may be reasonable to approximate the outgoing wavevector $k\beta_n$ by
  the negative incoming wavevector $k\beta_0$,
  or equivalently, to replace $\beta_n$ by $\beta_0=\sin\theta$.
  With this approximation, we have
  \begin{eqnarray}\sum_{n\in\mathbb{Z}} u_n\e^{\i n x} \approx \frac{-\e^{-\i k h(x)\sin(\theta)}}{\e^{\i k \beta_0 h(x)}}
  = -1 + 2\i k h(x)\beta_0 + \bigO(k^2|{h}|^2),\end{eqnarray}
 provided that the depth of the corrugation is small compared
 to the wavelength. Hence, we have
    \begin{eqnarray} u_n\approx v_n \triangleq \begin{cases}
       -1 & n = 0 \\ 2\i k \widehat{h}_n \beta_0 & n\neq 0
      \end{cases} \label{vn} \end{eqnarray}
 which  is sparse by  the sparseness assumption on $\{\widehat{h}_n\}$.
Let  $W= \sqrt{m} (v_n\e^{\i k \beta_n z_0})$. In view of \eqref{1.7} 
we have the estimate 
   \begin{eqnarray}
    \norm{X - X_s}_1 \approx \norm{X - W}_1
    \ = \ \sum_{n=-N/2}^{N/2-1}\sqrt{m} \abs{\e^{\i k \beta_n z_0}}\abs{u_n - v_n}
    \ \leq \ \sqrt{m}\epsilon_u
   \end{eqnarray}
   where $\epsilon_u \triangleq  \sum_{n=-N/2}^{N/2-1}\abs{u_n - v_n}$.
  The subsequent  numerical simulation shows that $u_n$ and $v_n$
   given in \eqref{vn} are indeed close to each other
   when the Rayleigh hypothesis is valid.

  \section{Numerical simulation}
  \subsection{Data synthesis}

We compute  the scattered  field $u(x,z_0)$ by the boundary integral method \cite{DEHM98,MACK00}.
  The scattered wave can be  represented by the Brakhage-Werner type ansatz, i.e.\ the representation via mixed  single-layer ($S$)  and double-layer ($K$) potentials
   \begin{eqnarray}u = (K-\i\eta S)\varphi \ \text{ on } \ \Omega_\Gamma\end{eqnarray}
with  a mixed layer density $\psi$ for
   a constant $\eta>0$ which can be adjusted to improve the condition number of the system.
   Explicitly, we can write
      \begin{eqnarray}
   \lefteqn{ u(x,z)}\label{33}\\
    &= &\int_{-\pi}^{\pi} \left( \frac{\partial }{\partial \nu'}\Phi \big((x,z) ,(x,h(x')) \big)
    -\i\eta \Phi \big((x,z) ,(x,h(x')) \big) \right)
     \psi(x') \cdot\sqrt{1+\dot{h}^2(x')}\d x'.\nonumber
   \end{eqnarray}
   %
Taking the limit $z\to h(x)^+$ and using the  properties
of single and double layer potentials, we obtain  the boundary integral equation  \cite{CK83,DEHM98}
\begin{eqnarray}
 -u^{\text{inc}}(x,h(x))
  &= & \frac{1}{2}\psi(x)  + \int_{-\pi}^{\pi} \left(
  \frac{\partial }{\partial \nu'}\Phi \big((x,h(x)) ,(x',h(x')) \big)\right. \nonumber\\
  &&\left. \hspace{2cm}
    -\i\eta \Phi \big((x,h(x)) ,(x',h(x')) \big) \phantom{\bigg|} \right)  \psi(x')
 \sqrt{1+\dot{h}^2(x')}
  \d x' \label{IntegralEq}\label{23}
  \end{eqnarray}
  (see Appendix). 
    Note that the integral in \eqref{IntegralEq} has weakly singular kernel, and the
    integral exists as an improper integral since the periodic Green's function
    \begin{eqnarray}\Phi (\vec{r},\vec{r}')
    =\frac{\i}{4}\sum_{n\in\mathbb{Z}}\e^{2\pi\i n k \cos\theta}
    H_0^{(1)}(k\abs{\vec{r}-\vec{r}'-2\pi n(1,0) })
    ,\ x,x'\in[-\pi,\pi)\end{eqnarray} has
    the same singularity as
    $H_0^{(1)}(x)\approx 1 + \frac{2\i}{\pi}(\ln\frac{x}{2}+\gamma)$
    where $\gamma\approx 0.5772$ is the Euler-Mascheroni constant.
   Moreover, 
    \begin{eqnarray}\frac{\partial }{\partial \nu'}H_0^{(1)}(k\abs{\vec{r}-\vec{r}'})
    =-k H_1^{(1)}(k\abs{\vec{r}-\vec{r}'})
    \frac{(\vec{r}-\vec{r}')\cdot\nu(\vec{r}')}{\abs{\vec{r}-\vec{r}'}}\end{eqnarray}
    converges to a finite limit (a curvature-like term w.r.t.\ the
    boundary) as $\vec{r}\to\vec{r}'$ \cite{CK98}
   implying the boundedness  of $\frac{\partial }{\partial \nu'}\Phi $ on $\Gamma$.

%

With  $\psi$ solved from \eqref{IntegralEq} and
 the Sommerfeld integral representation 
   \begin{eqnarray}
  \label{somm}
H^{(1)}_0(k|\vec{r}|)=\frac{1}{\pi}
  \int \e^{\i k (|z|\beta+x\alpha)} \frac{\d\alpha}{\beta},  \end{eqnarray}
  where 
  \begin{eqnarray}
  \label{31}
  \beta=\left\{\begin{array}{ll}
  \sqrt{1-\alpha^2}, & |\alpha|<1\\
  \i\sqrt{\alpha^2-1},& |\alpha|> 1
  \end{array}
  \right.
  \end{eqnarray}
we obtain  from \eqref{33}  the outgoing wave expansion
for the scattered field
   \begin{eqnarray}
  u(x,z)
   =\sum_{n\in\mathbb{Z}}\e^{\i k (\alpha_n x + \beta_n z)}
   \left(\frac{1}{4\pi}\int_{-\pi}^{\pi}
   \e^{-\i k\big(\alpha_n x' + \beta_n h(x')\big)} g_n(x')
   \psi(x')\d x' \right)   ,\quad   z> h_{\rm max} \label{25}
   \end{eqnarray}
with 
   \begin{eqnarray}
   g_n(x') = k - k\dot{h}(x')\frac{\alpha_n}{\beta_n} +\frac{\eta}{\beta_n} {\sqrt{1+\dot{h}^2(x')}}.\label{28}
   \end{eqnarray}
 Comparing \eqref{25} with 
   \eqref{1.4} we arrive at  the expression 
   \begin{eqnarray}
   u_n &= \frac{1}{4\pi}\int_{-\pi}^{\pi}\e^{-\i k\big(\alpha_n x' + \beta_n h(x')\big)}
   g_n(x') \psi(x')\d x'  \label{26}
   \end{eqnarray}
   relating the angular spectrum of the scattered field to
   the mixed layer density $\psi$.

\eqref{23} and \eqref{26} motivates the following
iterative reconstruction scheme. Given $h^{(m)}$, $m=1,2,3,\ldots$
first solve for $\psi^{(m)}$ from  
\begin{eqnarray}
 -u^{\text{inc}}(x,h^{(m)}(x))
  &= & \frac{1}{2}\psi^{(m)}(x)  + \int_{-\pi}^{\pi} \left(
  \frac{\partial }{\partial \nu'}\Phi \big((x,h^{(m)}(x)) ,(x',h^{(m)}(x')) \big)\right. \nonumber\\
  &&
    -\i\eta \Phi \big((x,h^{(m)}(x)) ,(x',h^{(m)}(x')) \big)  \Big)  \psi^{(m)}(x')
 \sqrt{1+|\dot{h}^{(m)}|^2(x')}
  \d x' \label{29}
  \end{eqnarray}
  and then solve for $h^{(m+1)}$ from 
     \begin{eqnarray}
   u_n &=& \frac{1}{4\pi}\int_{-\pi}^{\pi}\e^{-\i k\big(\alpha_n x' + \beta_n h^{(m)}(x')\big)}
   g^{(m+1)}_n(x') \psi^{(m)}(x')\d x',\quad n\in \mathbb{Z} \label{30} \\
   g^{(m+1)}_n(x') &= &k - k\dot{h}^{(m+1)}(x')\frac{\alpha_n}{\beta_n}+
   \frac{\eta}{\beta_n}{\sqrt{1+|\dot{h}^{(m)}|^2(x')}}. 
\end{eqnarray}
Note that both \eqref{29} and \eqref{30} are linear equations. 

A natural candidate for the initial guess of  the above iteration is
the one obtained under the Rayleigh hypothesis that the validity
of \eqref{25} is extended to the region $z> h(x)$. Specifically, we
extend  \eqref{25} all the way to boundary and study the nonlinear equation \eqref{etbc}. Indeed, this alone
produces excellent results for shallow corrugations and will
be the focus of the following numerical experiments. 
Scattering and imaging of shallow corrugations can also
be treated by assuming the Born approximation \cite{DPT}.

  In our numerical simulations, we set $128$ nodes to solve the
  boundary integral equation \eqref{23}  by the Nystr\"om method, with $\eta = 1$. Figure \ref{fig3} shows two examples of the  computed scattered field.  
We define  $R_a(h)\triangleq \max_n 2 |n\widehat{h}_n|$ for a rough metric of the validity of the Rayleigh hypothesis.

   \begin{figure}\centering
        \includegraphics[width=3in]{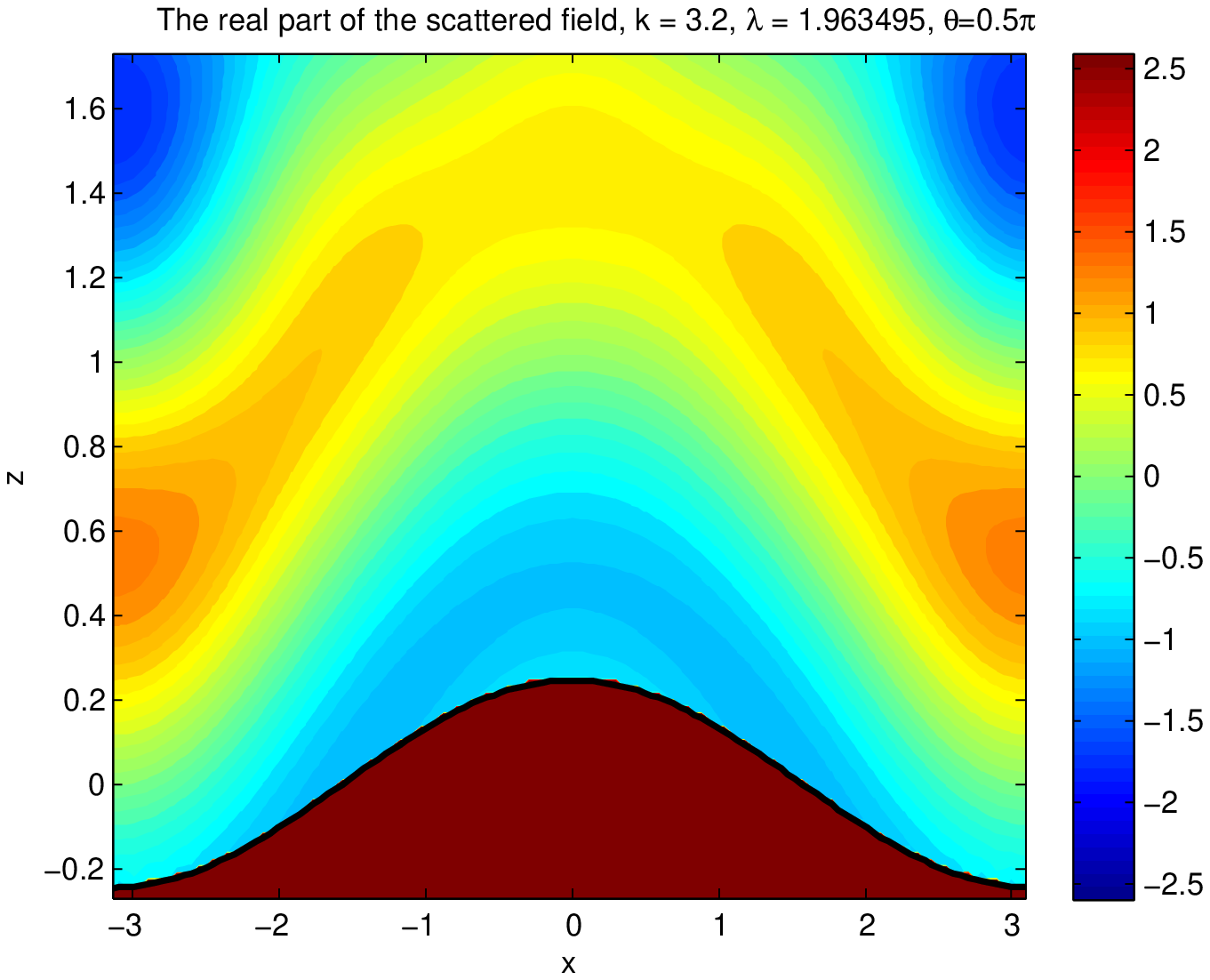}
     \includegraphics[width=3in]{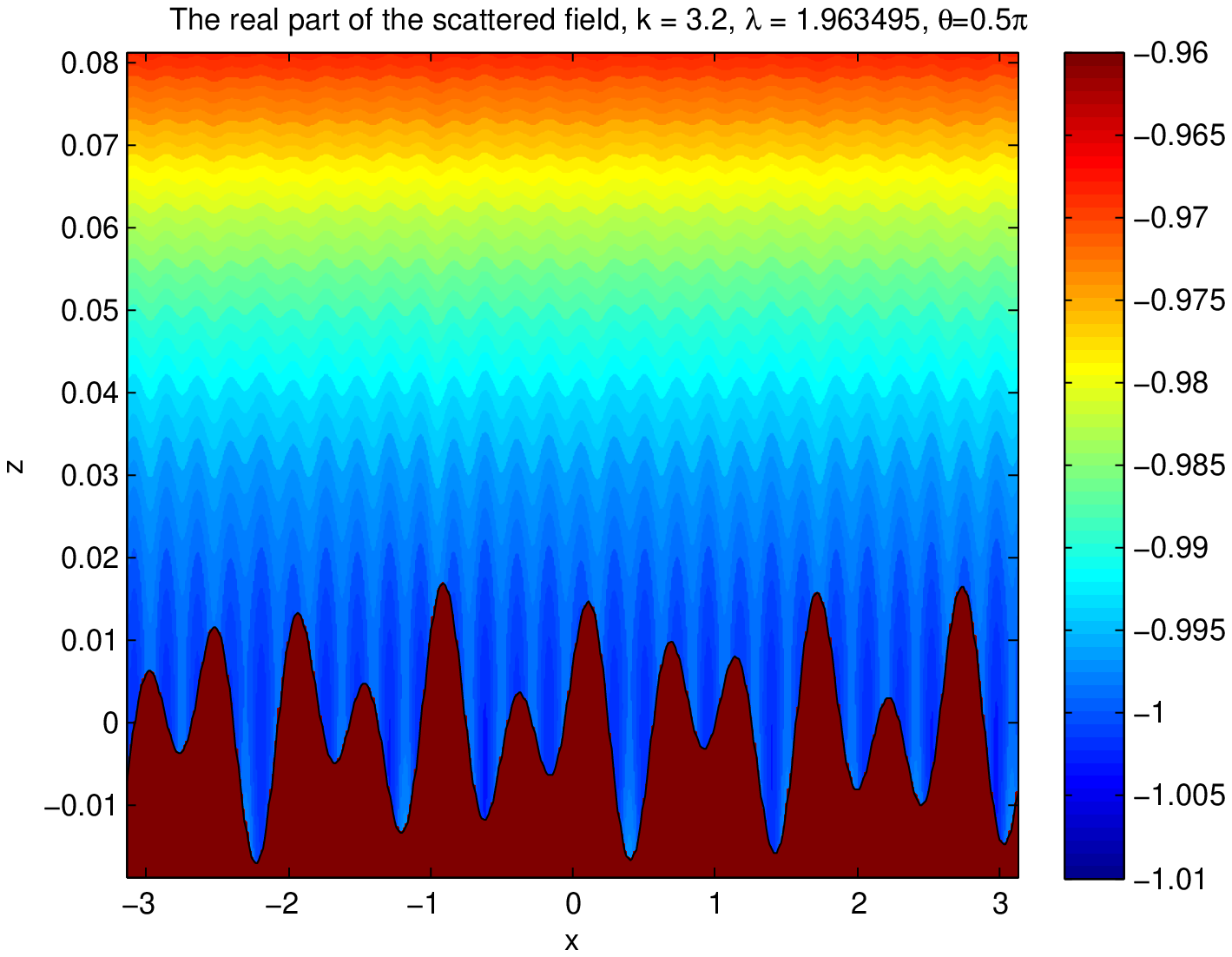}\\
     \caption{The real part of the scattered field induced by the profile
     $h(x)=0.2454\sin(x)$ (left)
    and the profile  $h(x)=0.01\sin(12x)+0.007\cos(7x)$ (right)
   with $L=2\pi$  and vertical incident wave.}
     \label{fig3}
  \end{figure}

  \subsection{Surface reconstruction}
 Solve for $h(x)$ from \eqref{etbc}, we consider  the following three algorithms: the first
  two  are pointwise matching schemes and the third
  is a global fitting scheme. 

  \begin{enumerate}
   \item Point-wise, fixed-point iteration for $h(x), \forall x\in [-\pi, \pi)$. A fixed point iteration algorithm was introduced in
     \cite{GV93,RGC81} and is described below. The initial condition $h^{[0]}(x)$
     is chosen in the following way. For $\theta\approx\frac{\pi}{2}$
 the angular spectrum   $\{u_n\}$ is
   concentrated at $n=0$. Substituting  $\beta_n$ by $1$ in \eqref{etbc} yields 
  \begin{eqnarray}
  \label{22}
  \e^{\i k h(x)}\cdot\sum_{n\in\mathbb{Z}}u_n
   \e^{\i k \alpha_n x} = -\e^{\i k \big(x\cos\theta - h(x)\sin\theta\big)}.
   \end{eqnarray}
 One solves \eqref{22} by  the iterative scheme
   \begin{align}
    h^{[0]}(x) &= \frac{ \ln \left(-\sum_{n\in\mathbb{Z}}u_n
    \e^{\i k\alpha_n x}  \right) }
    {-2\i k} ,  \label{ig} \\
    h^{[n+1]}(x) &= \frac{ \ln \left(-\sum_{n\in\mathbb{Z}}u_n
    \e^{\i k\big(\alpha_n x + (\beta_n-1)h^{[n]}(x)\big)}  \right) }
    {-2\i k }
   \end{align}
   for  $ n=1,2,\ldots $ and all $x\in[-\pi,\pi)$.
   \item Newton's method. From \eqref{etbc}, for each $x\in [-\pi,\pi)$ we set
    \begin{eqnarray} e(h;x) = \e^{\i kx_j\cos\theta}\e^{-\i k h\sin\theta} +
     \sum_{n\in\mathbb{Z}}u_n \e^{\i k \alpha_n x_j} \e^{\i k \beta_n h} ,
     \label{littlee}
    \end{eqnarray}
    and solve $e(h,x)=0$  for $h(x)$ by Newton's method
    \begin{eqnarray}h^{[i+1]} = h^{[i]}
    - \frac{e(h^{[i]};x)}{\frac{\d}{\d h}e(h^{[i]};x)} \end{eqnarray}
    with initial value \eqref{ig}.
   \item Nonlinear least squares fitting. Let
    \begin{eqnarray}F(\vec{a}) = \norm{e( \sum_n a_n\phi_n(\cdot),\cdot)}^2 =
    \sum_j |e(\sum_n a_n\phi_n(x_j),x_j)|^2 ,\end{eqnarray}
    where $e(h;x_j)$ is defined in \eqref{littlee} and
    $\vec{a}=(a_1,a_2,\dots)$ is the vector of the coefficients of
    an expansion of $h$ corresponding to a frame $\{\phi_n\}$, i.e., $h = \sum_n a_n \phi_n$.
    Consider minimizing the nonlinear least square
    \begin{eqnarray}\min_{\vec{a}} F(\vec{a}).\end{eqnarray}
    The basis function are chosen to be $\sin(nx)$ and $\cos(nx)$ for $n\in \Pi\subset\mathbb{N}$, where the index
    set $\Pi$ contains those indices $n$ such that $|\tilde{u_n}|$ are relatively large.
    The Matlab subroutine \verb"lsqnonlin" is applied, which is based on the subspace trust region method.
  \end{enumerate}

  \subsection{Examples}
  In the following examples we apply vertical incident wave,
  $\theta=\frac{\pi}{2}$, with the wave number $k=3.2$ (i.e., the wavelength $\lambda\approx 1.9635$).
  After synthesizing $u(x_j,z_0)$ for $j=1,2,\ldots,m$, $1\%$ additive noise, with respect to
  $\|u(\cdot,z_0)\|_2/\sqrt{m}$, is added to the measured scattered field $u$.
The profile functions  $h$ are $L=2\pi$ periodic, defined in one period $[-\pi,\pi)$, and periodically extended into $\mathbb{R}$.
  The
  bound for the exponential factor is set to be $C_e = \log(25)$.
  Yall1 \cite{Yall1} algorithm, a Basis Pursuit solver,
  is applied to solve \eqref{CSsystem} for vector $\tilde{X}$. To avoid exponential amplification of
  small components of $\tilde{X}$, we apply a threshold level $\tau = 10\%\cdot \max_{n\neq 0}
  |\tilde{X}_n|$ and filter
  out the components  below $\tau$, then compute $\tilde{u_n}$ from \eqref{Xtoun}.


  In Figures \ref{fig4}-\ref{fig10},  the right panels show  the exact profile $h(x)$
  (black solid line), and the reconstruction under the three algorithms: Newton's method ``\verb"Newton"'',
  fixed point iteration ``\verb"Fixed pt iter"'', and nonlinear least squares fitting
  ``\verb"NLS fit"''. The length of the black strip on the top
  of each plot  indicates  the
  wavelength, and its height indicates the vertical coordinate $z_0$ of the sampling points. The left and middle panels show
the real (left) and imaginary parts (middle)
 the angular spectrum $u_n$ (blue crosses), the estimated angular spectrum $\tilde{u_n}$ (red dots), and the theoretical estimate
  $v_n$ (green circles). Note  the different scales for the
  real and imaginary parts in Figures \ref{fig4}-\ref{fig7} where
  $R_a$ is relatively small. This is no longer the case
  in Figures \ref{fig9} and \ref{fig10} for which the Rayleigh
  hypothesis is known to be false.


  \begin{figure}\centering
    \includegraphics[width=3.2in]{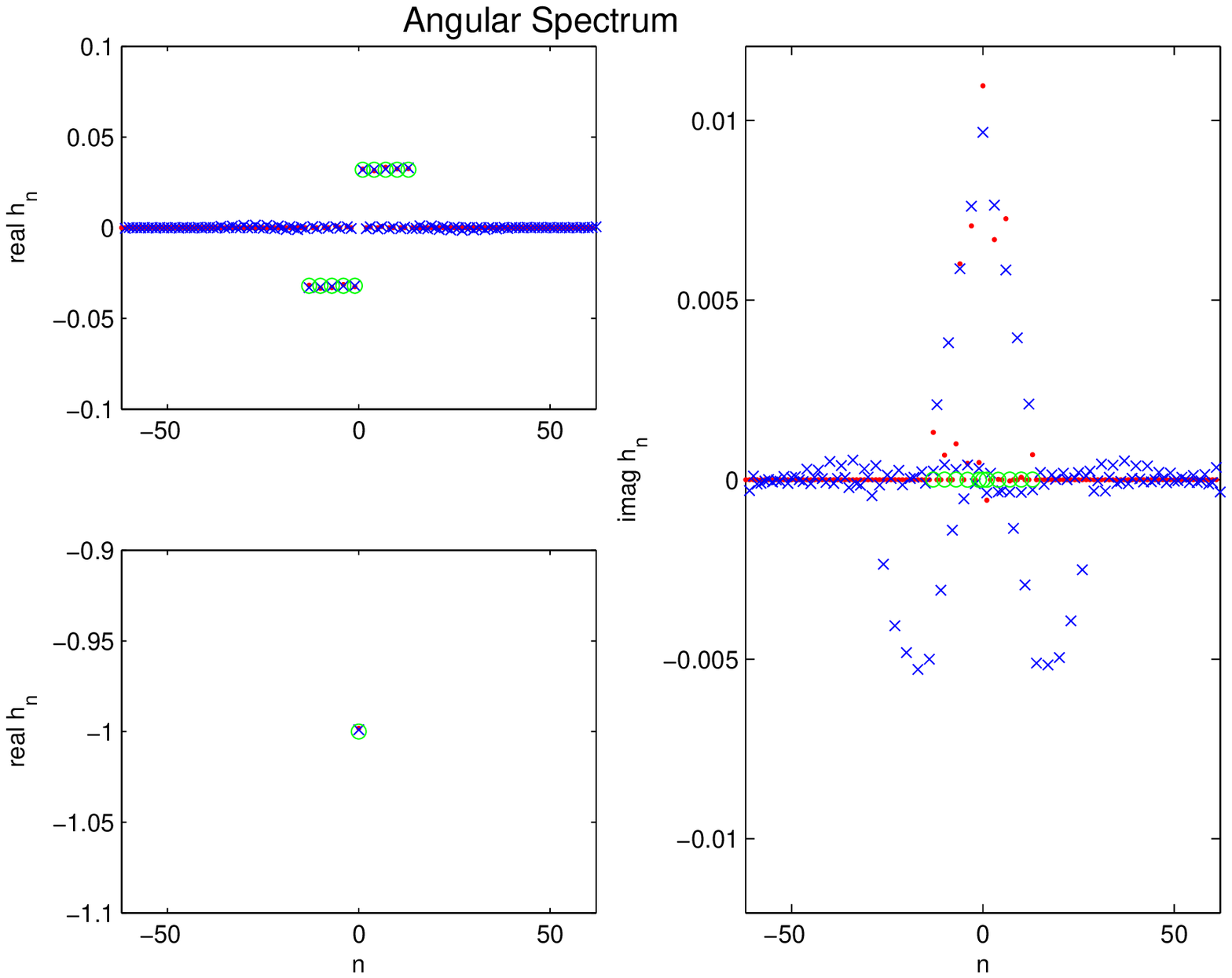}\
      \includegraphics[width=3.2in]{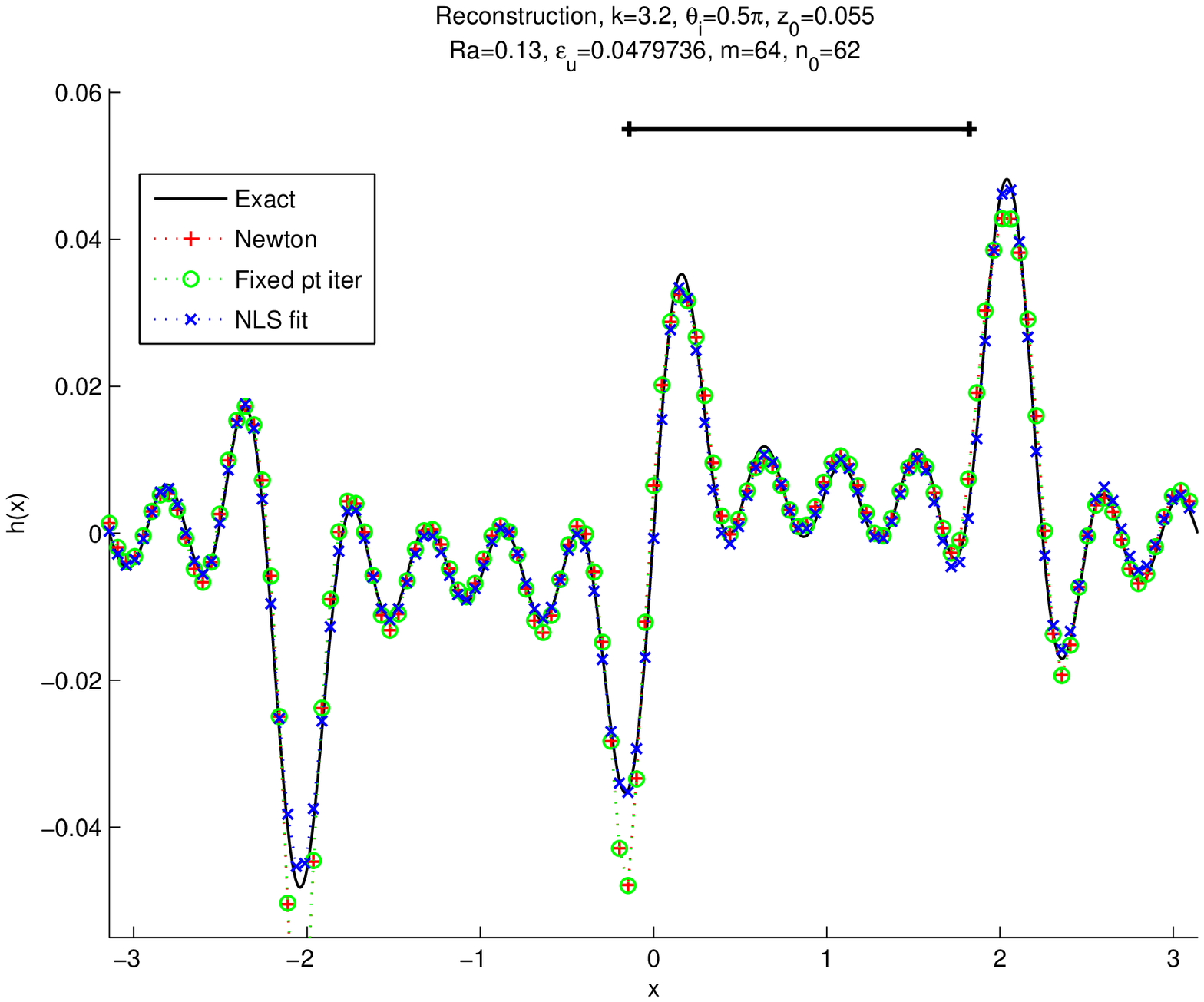} 
    \caption{The profile of 5 Fourier modes $h(x)=0.01\big(\sum_{p=0}^4 \sin((1+3p)x)\big)$ and the reconstructions.}
    \label{fig:ST1}
    \label{fig4}
  \end{figure}


  \begin{figure}\centering
   \includegraphics[width=3.2in]{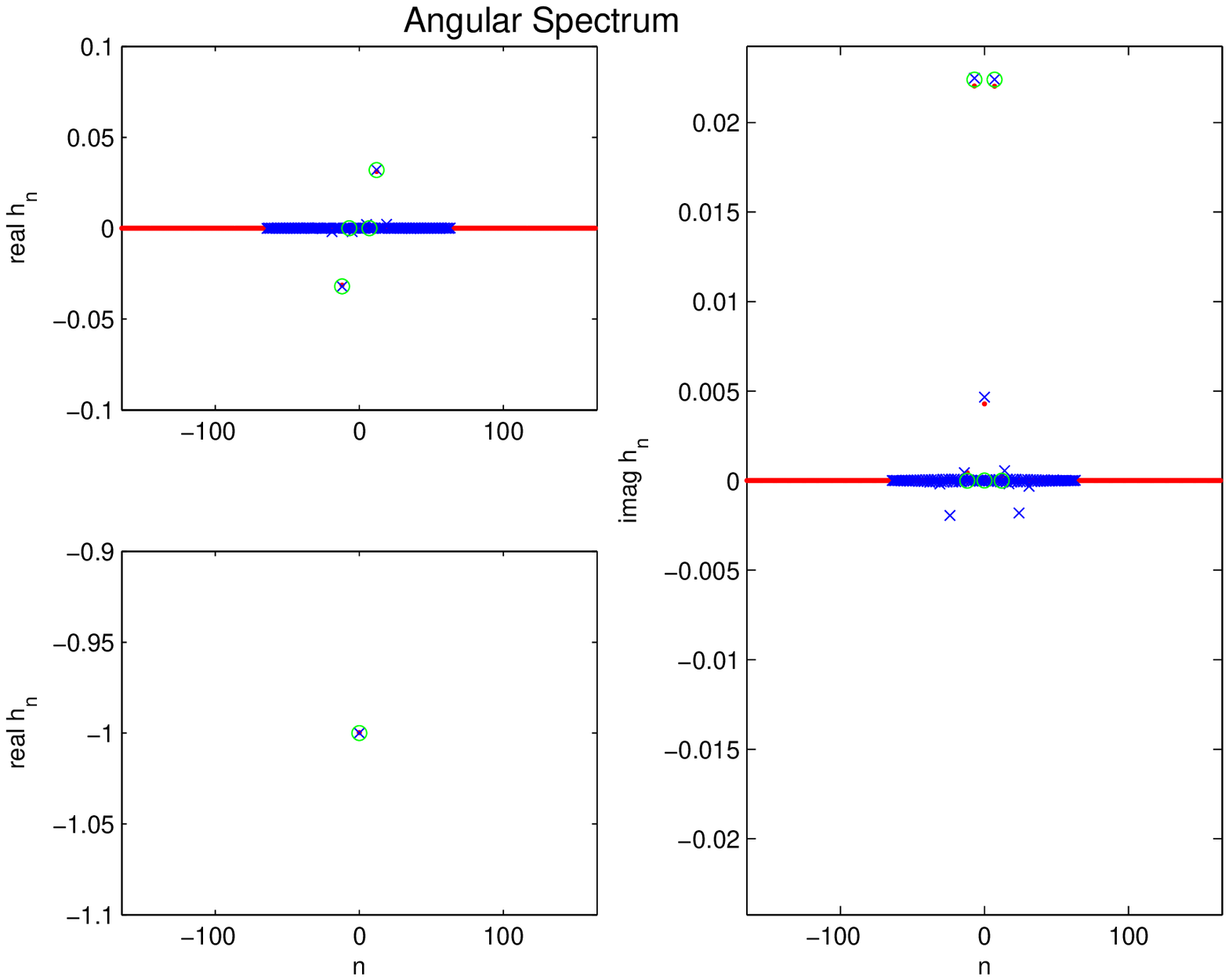}
    \includegraphics[width=3.2in]{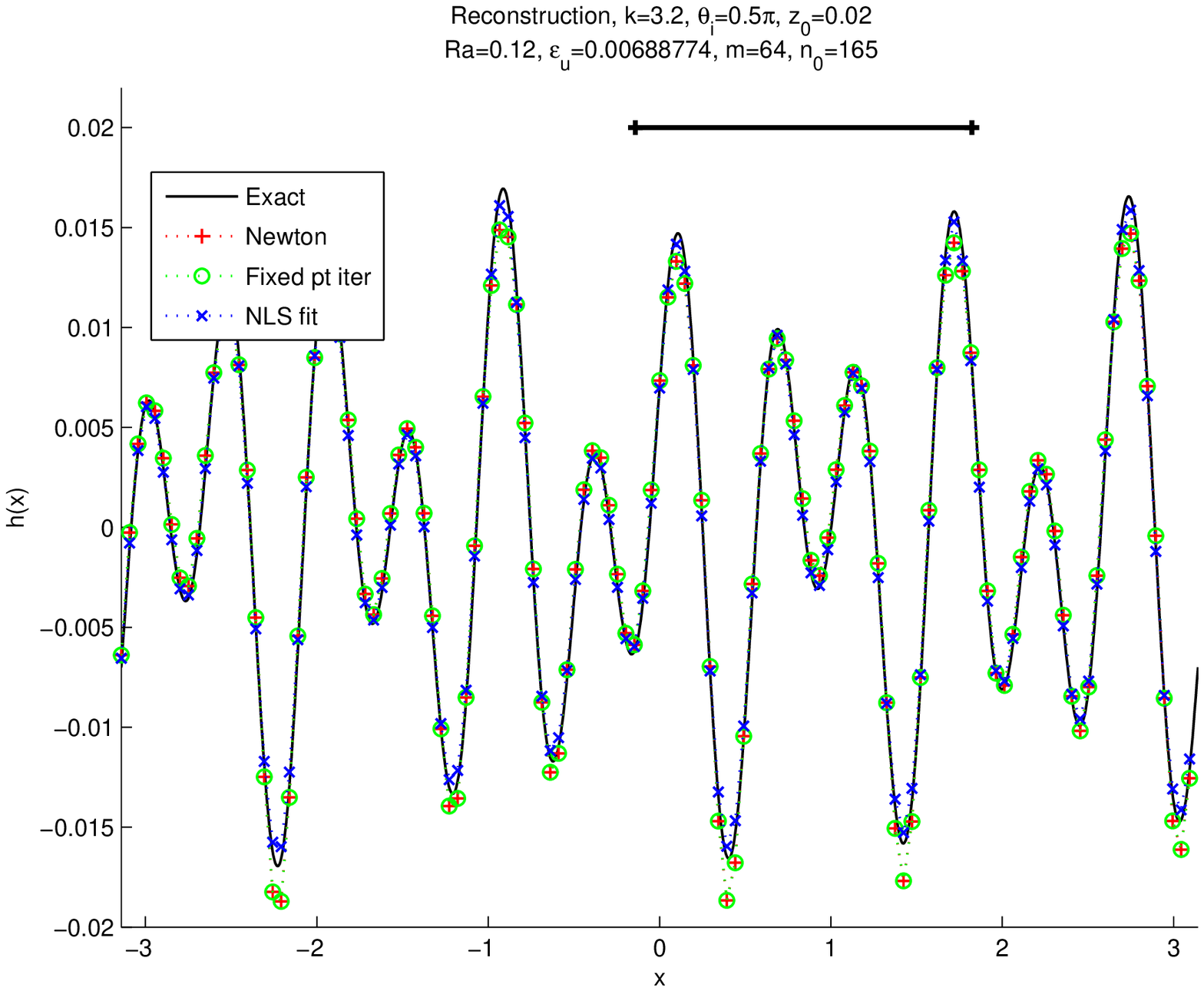} 
    \caption{Profile of two Fourier modes  $h(x)=0.01\sin(12x)+0.007\cos(7x)$ and the reconstructions.}
    \label{fig:ST2}
    \label{fig5}
  \end{figure}


  \begin{figure}\centering
    \includegraphics[width=3.2in]{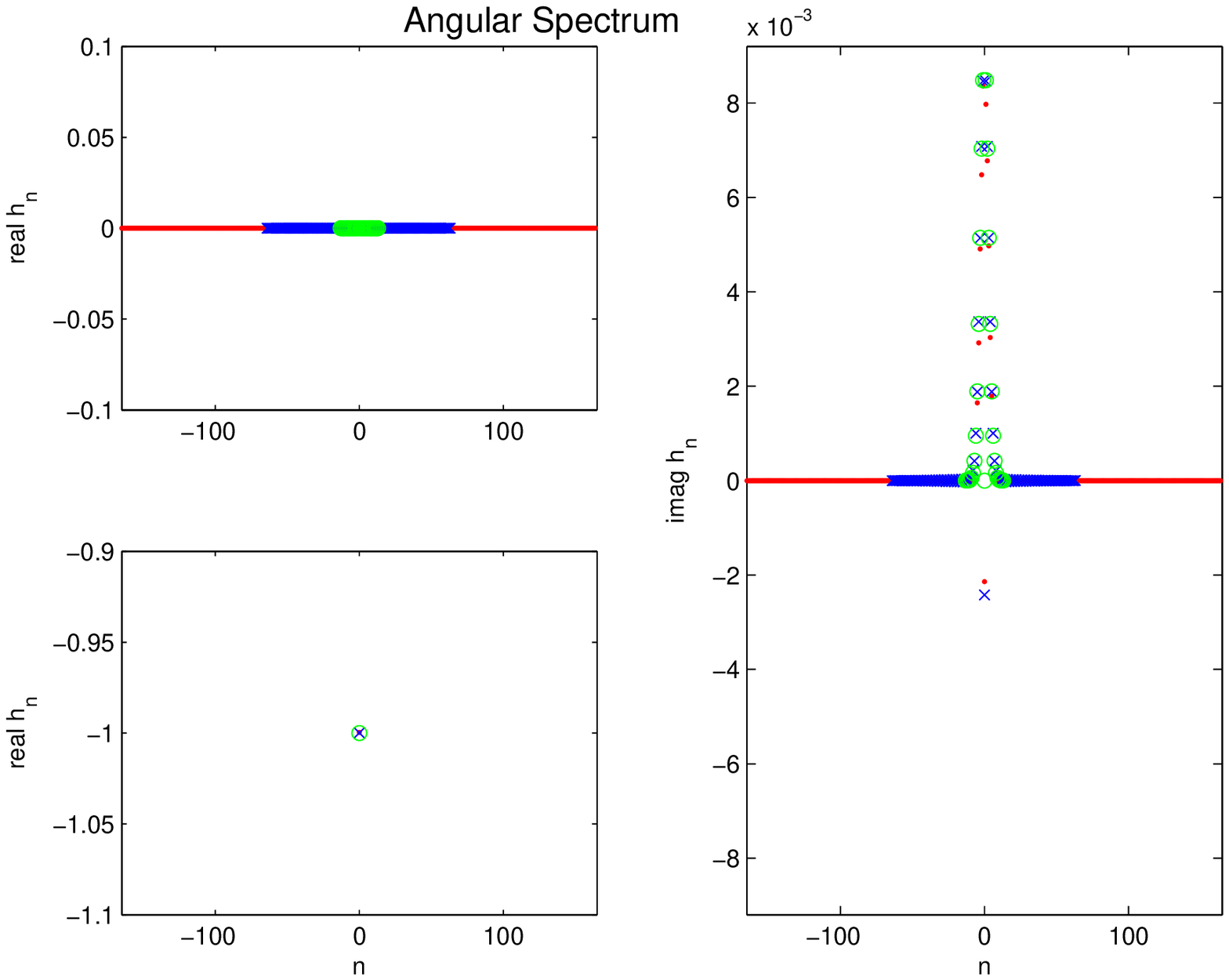}
      \includegraphics[width=3.2in]{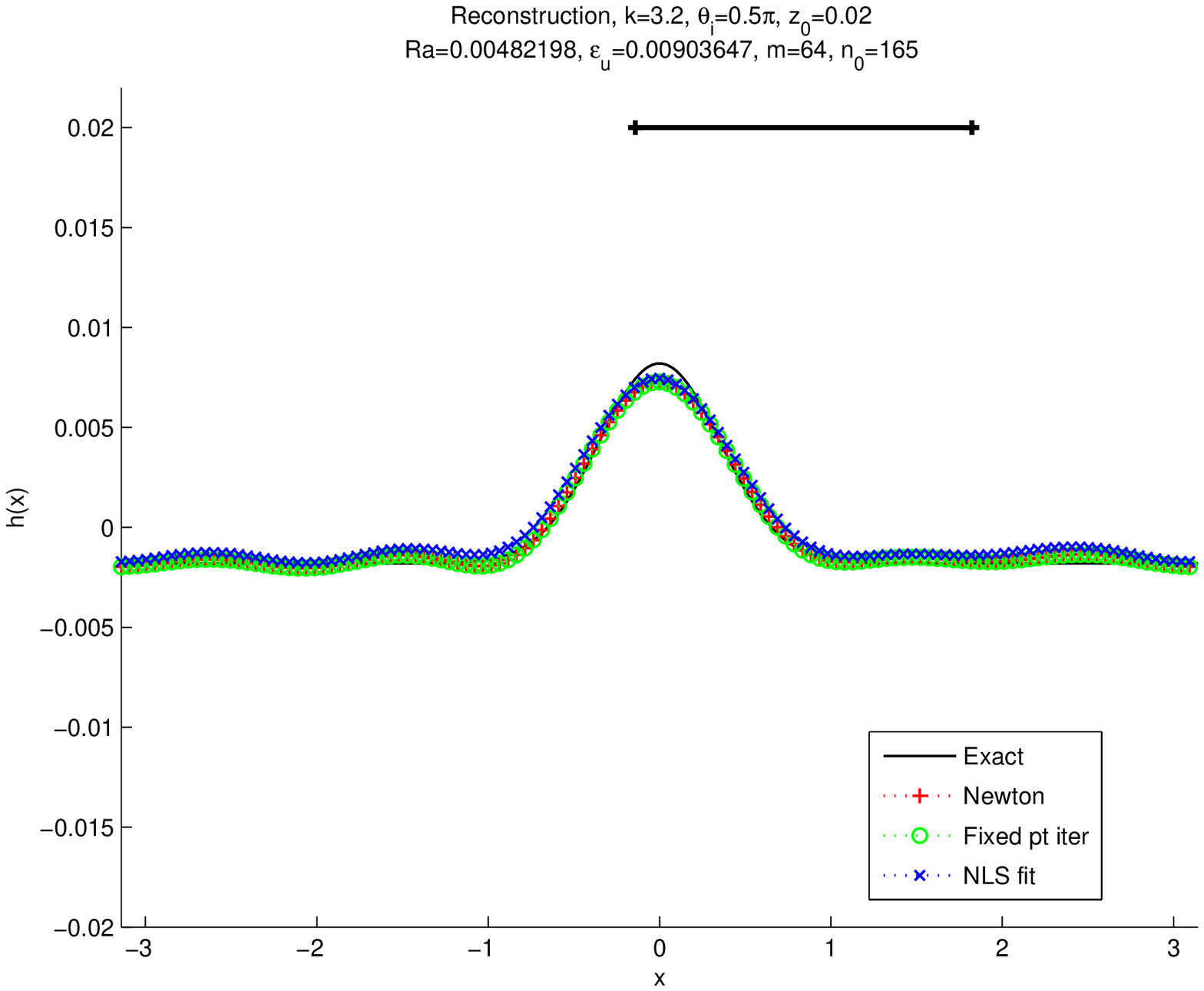}
    \caption{Periodized Gaussian
    $h(x)= b\big(\e^{-(ax)^2} -  \frac{\erf(a\pi)}{\pi\sqrt{\pi}} \big)\cdot\tilde{\chi}_{[-0.9\pi,0.9\pi]}(x)$,
    $a=2$, $b=0.01$ and the reconstructions. 
    Here $\tilde{\chi}$ is a smoothed indicator function.}
    \label{fig:PG}
    \label{fig6}
  \end{figure}


  \begin{figure}\centering
      \includegraphics[width=3.2in]{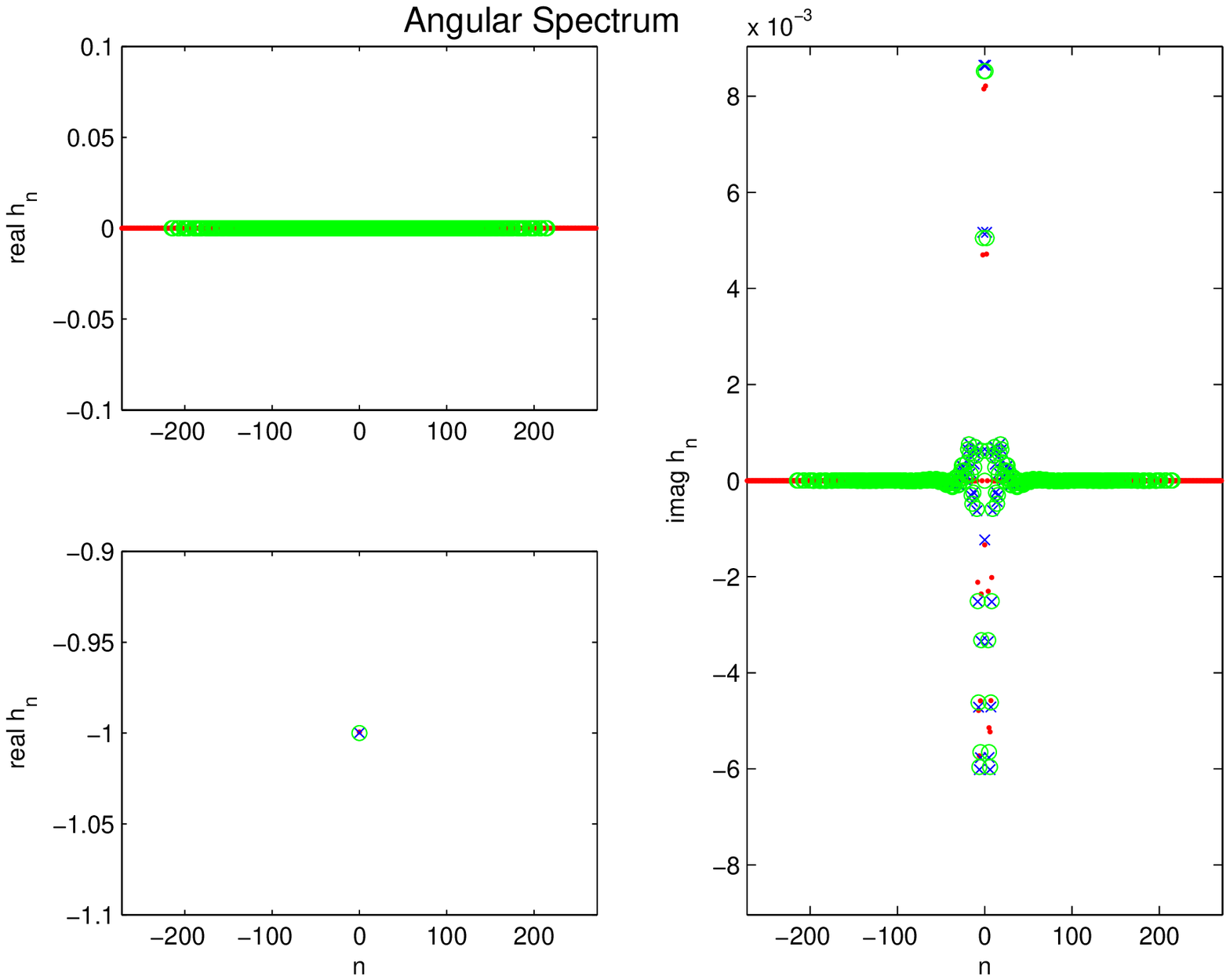}
    \includegraphics[width=3.2in]{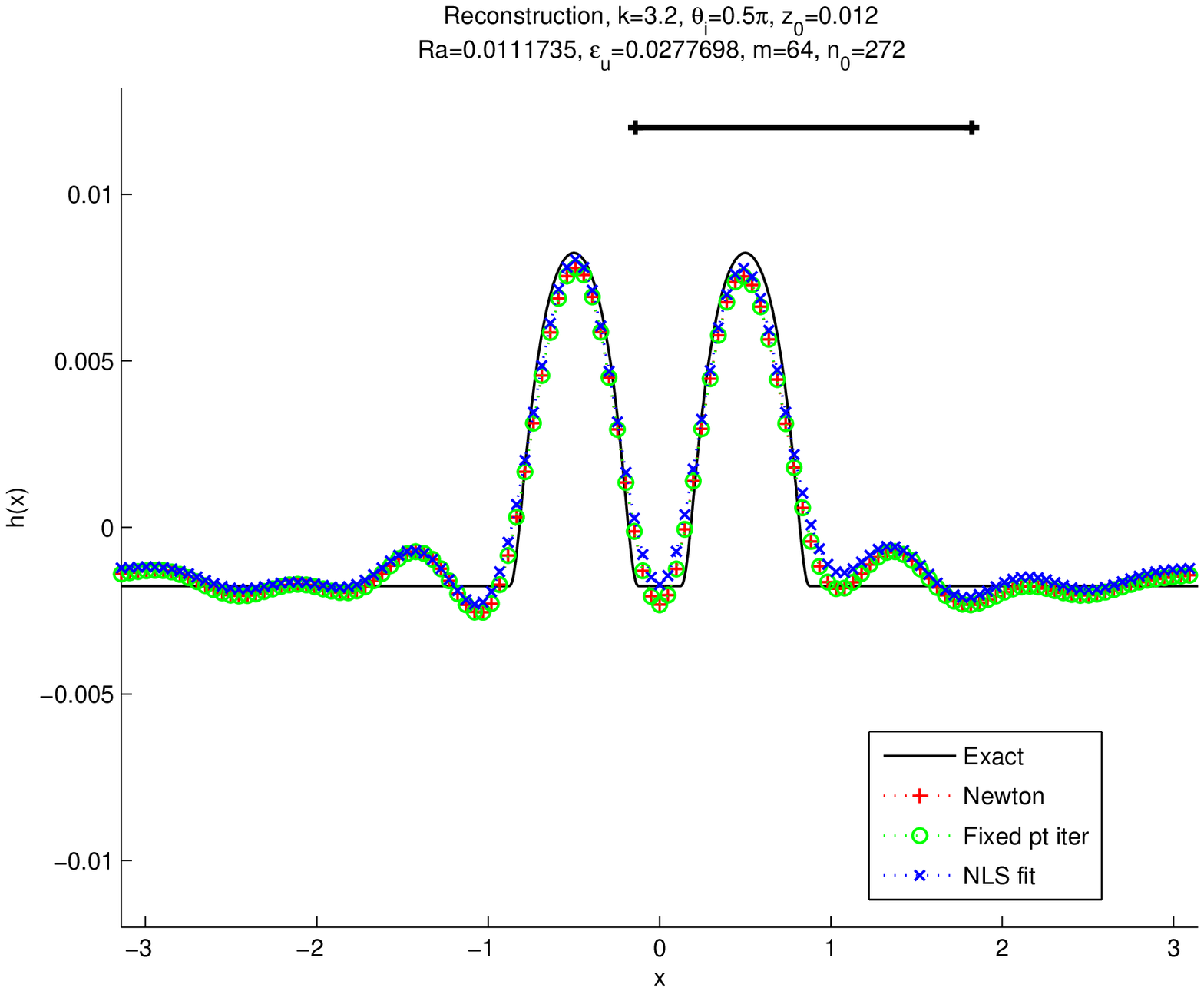}
    \caption{Double subwavelength peaks
    $h(x)= b\big( \zeta(a(x-\frac{1}{2})) + \zeta(a(x+\frac{1}{2})) \big)$, $a=2.5$, $b=0.01$,
    $\zeta(x)=\exp\left(1-\frac{1}{x^2-1}\right)  \chi_{(-1,1)}(x) + c_0$ 
    and the reconstructions. Here the constant
    $c_0$ is chosen such that
    $\widehat{\zeta}_0=0$.}
    \label{fig:Peaks}
    \label{fig7}
  \end{figure}



  \begin{figure}\centering
   \includegraphics[width=3.2in]{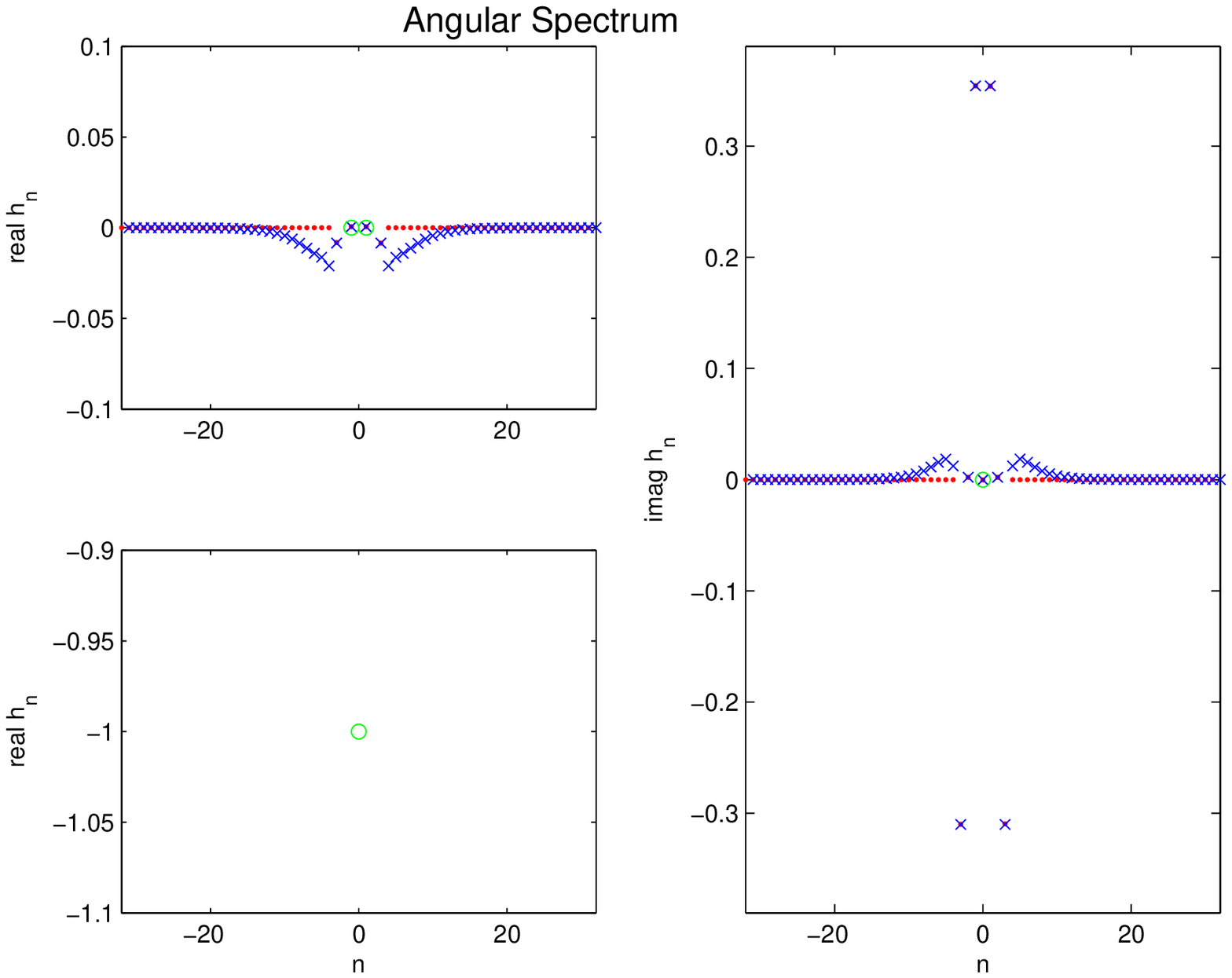}
       \includegraphics[width=3.2in]{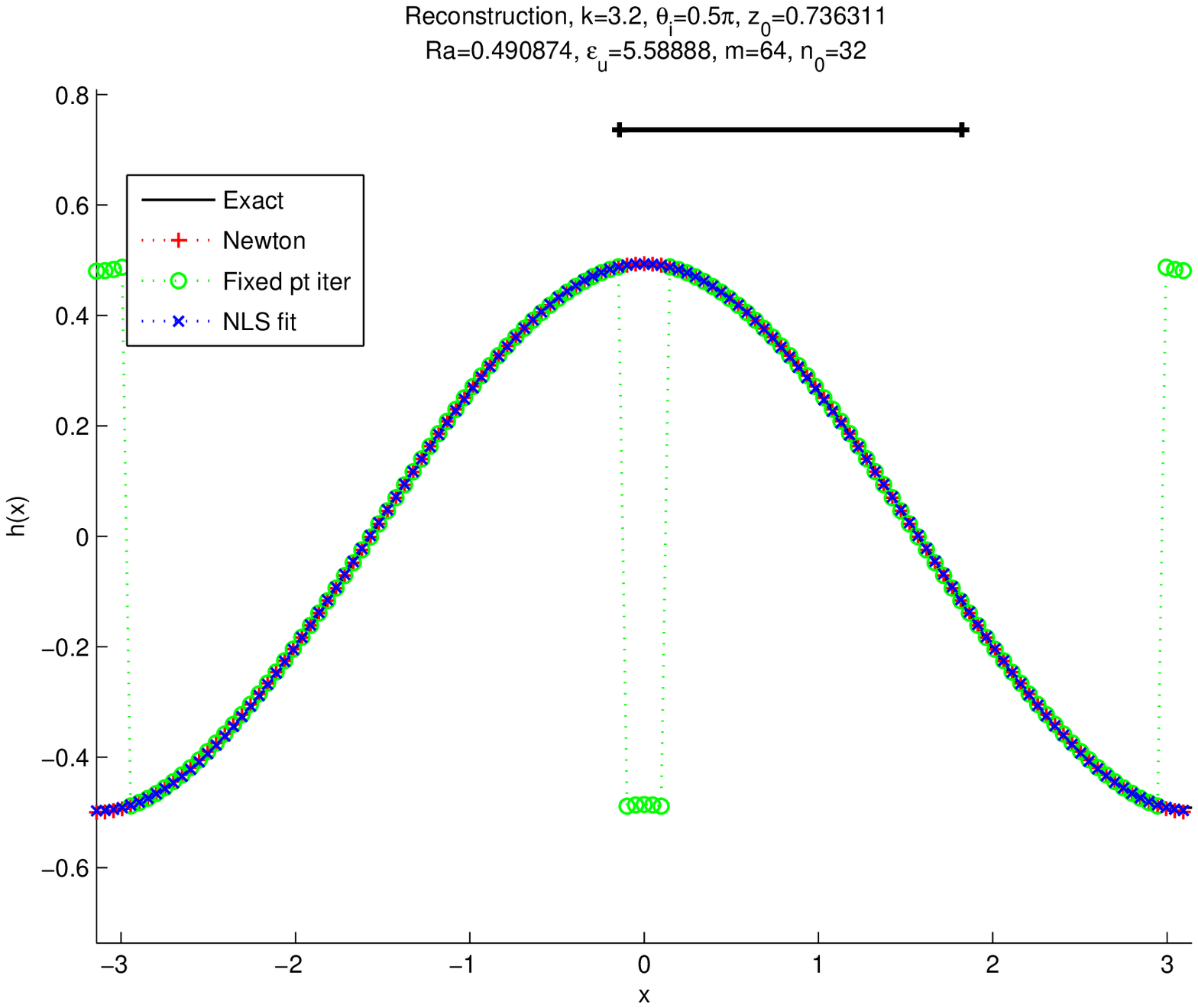}
    \caption{
    $h(x)=0.491\cos(x)$ }
    \label{fig:sin_a1b1}
    \label{fig9}
  \end{figure}

  \begin{figure}\centering
     \includegraphics[width=3.2in]{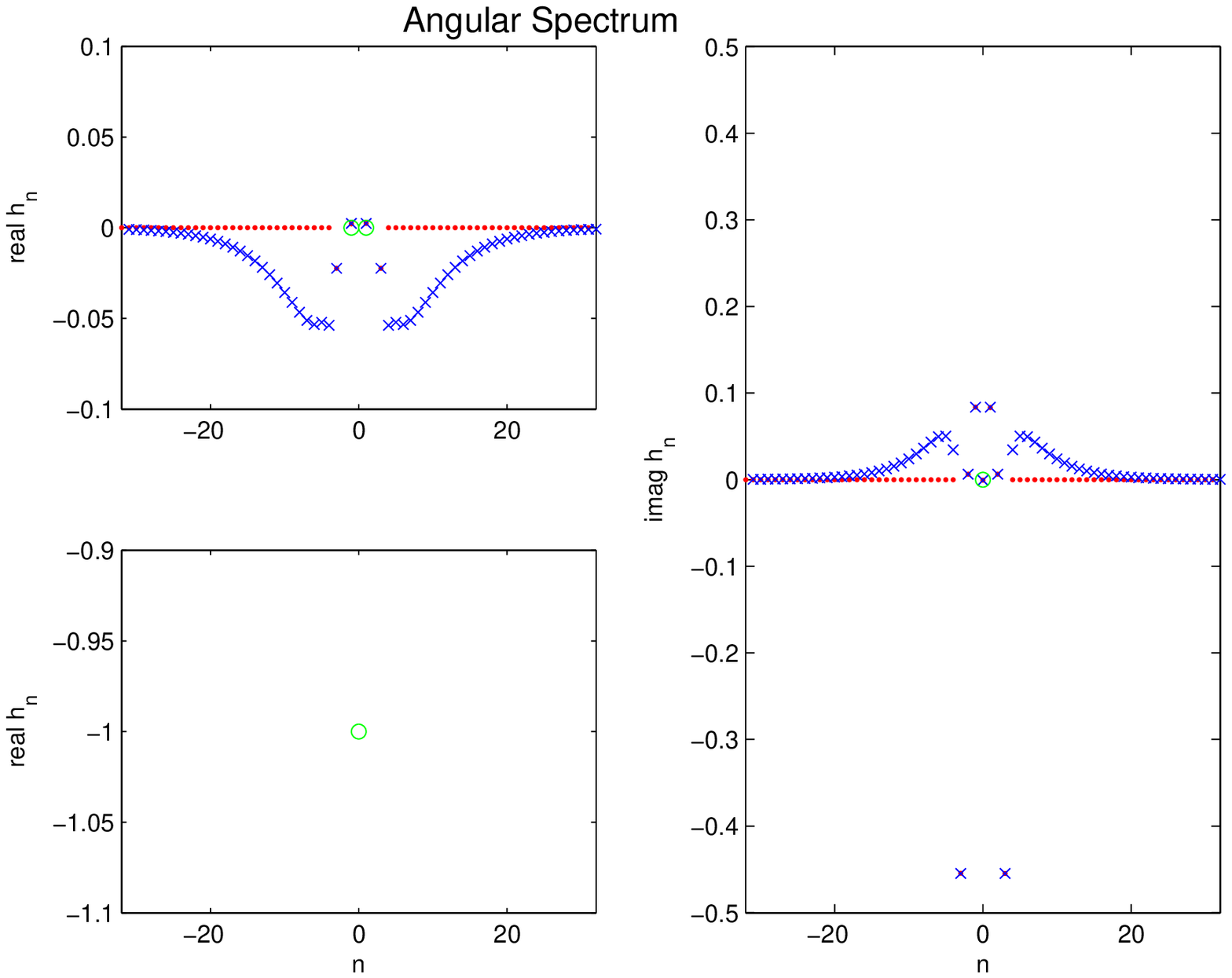}
        \includegraphics[width=3.2in]{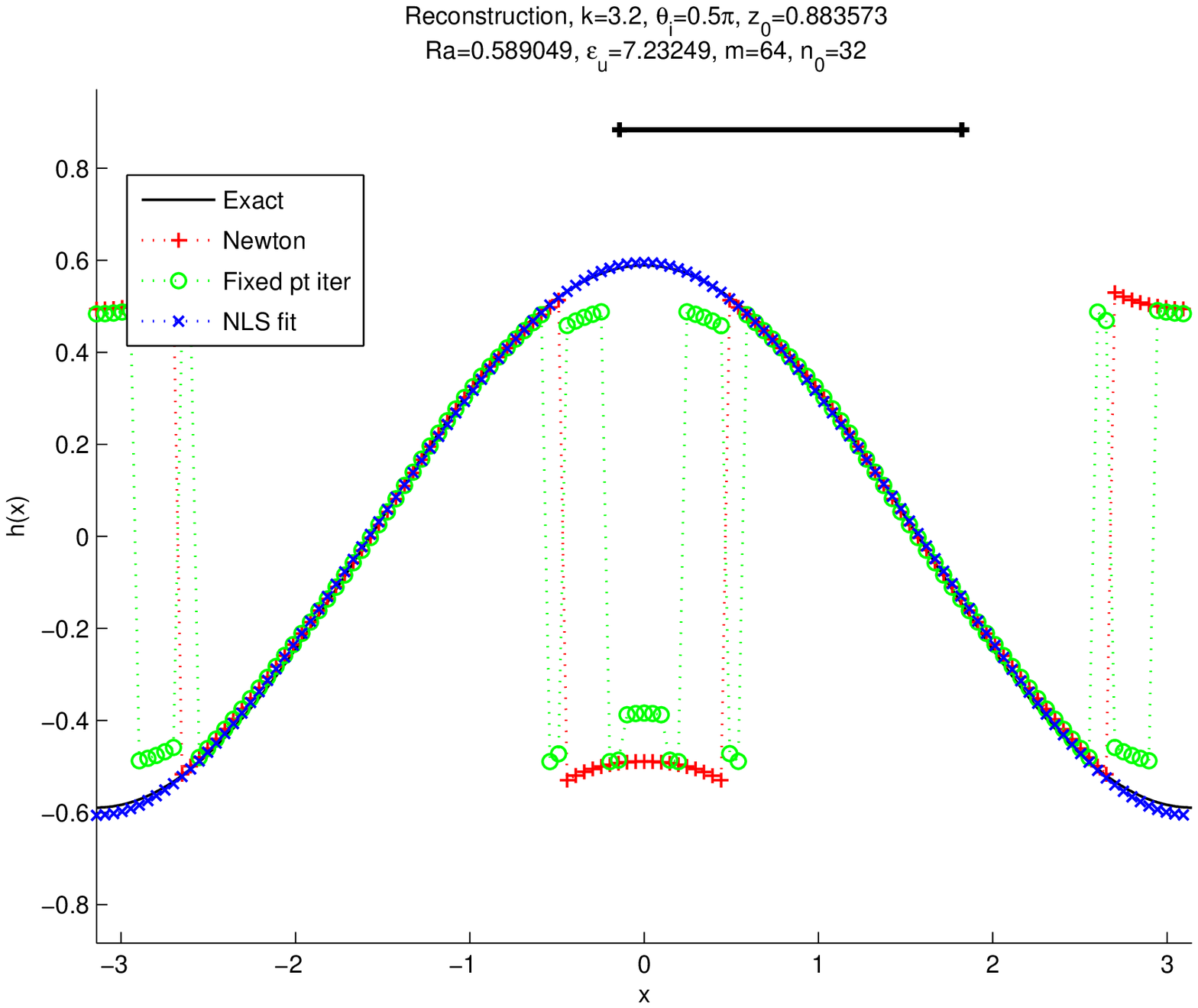}
    \caption{
    $h(x)= 0.589\cos(x)$.}
    \label{fig:sin_a1b12}
    \label{fig10}
  \end{figure}



   Figures \ref{fig4} and \ref{fig5} show the results for
   profiles $h(x)$ with sparse Fourier coefficients.
   The prediction
   $v_n$ captures well the dominant  component $\Re[u_0]$ and so does the sparse reconstruction $\tilde{u_n}$ the  other significant components of the angular spectrum. For reconstruction (right panels)
   the nonlinear least squares is the best performer
while the pointwise iterative
   methods may produce visible  undershoots at the peaks and troughs.

   For the  Gaussian profile (Fig \ref{fig:PG}) and subwavelength double-peaks (Fig \ref{fig:Peaks}), again
   $v_n$ captures well the dominant  component $\Re[u_0]$ and so does the sparse reconstruction $\tilde{u_n}$ most other significant components of the angular spectrum.    
   The angular spectrum for the latter case occupies a wider range of modes than the former case since the two  peaks are sharper than
   the Gaussian.  As a consequence,
   the reconstruction is more accurate in the former case.
   For the latter case, all three reconstructions undershoot
   the peaks and produce  fluctuations at the flat part of the profile.

   Figures  \ref{fig:sin_a1b1} and \ref{fig10} are the results for simple sinusoids  when the Rayleigh hypothesis is known to fail (${ab}>0.448$). The failure of the Rayleigh hypothesis manifests in the 
broadening of the support of the angular spectrum. Furthermore, the imaginary part of
the angular spectrum  
is  order of magnitudes larger than those in Figures \ref{fig4}-\ref{fig7}.    As a result, the angular spectrum $\{u_n\}$ is less compressible  and  not well recovered by the compressed
   sensing techniques.  In 
   both cases,
 the simple prediction $v_n$ fails to
  capture  the  dominant components of the angular spectrum.

   Nevertheless,  the nonlinear least squares fitting provides
   an accurate  reconstruction of the profile in both cases.
   The Newton iteration  converges in Figure \ref{fig9} but
   fails near the peaks and
troughs in Figure \ref{fig10} while
the fixed point iteration fails to converges near the peaks and
troughs in both cases.
When $ab$ is further increased (to, e.g. $0.736$),
then all three methods fail to recover the profile.

   \section{Conclusion}
We have proposed a compressed sensing scheme for near-field imaging of corrugations of relative sparse Fourier components.
The scheme employs  random sparse measurement  of near field
to recover the angular spectrum of the scattered field. We have shown heuristically and numerically that under the Rayleigh hypothesis
the angular spectrum is indeed sparse or compressible and
amenable to compressed sensing techniques.  

 We then
develop  iteration schemes for recovering the  surface profile from the angular
spectrum.  
Specifically, under the Rayleigh hypothesis we 
have tested three iterative schemes. The nonlinear least squares in
the Fourier basis
has the best performance among the three and produces
accurate reconstructions even when the Rayleigh hypothesis
is known to be false. 

The full iteration scheme \eqref{29}-\eqref{30} beyond the limitation of
the Rayleigh hypothesis will require non-sparse measurements 
for the angular spectrum data
and will be studied elsewhere. 

\bigskip

\noindent {\bf Acknowledgement.}   The research supported in part by NSF Grant DMS 0908535. 

\bibliographystyle{amsalpha}

\appendix
\section{Derivation of the boundary integral equation \eqref{23}}
   The term $\frac{1}{2}\psi(x)$ in \eqref{IntegralEq} arises due to the jump discontinuity
    for the double layer potential across the boundary, whereas the single layer potential is continuous.
    More specifically, let 
    \begin{align}
     u_S(\vec{r}) &= \int_{\Gamma} \Phi (\vec{r},\vec{r}') 
     \psi_S(\vec{r}') \d S(\vec{r}') \\
     u_D(\vec{r}) &= \int_\Gamma  \frac{\partial }{\partial \nu'}\Phi (\vec{r},\vec{r}') 
     \psi_D(\vec{r}') \d S(\vec{r}')
    \end{align}
   be  the single and double layer potentials respectively for
   $\vec{r}=(x,z) \in\mathbb{R}^2\setminus\Gamma$. Furthermore we denote
    $\vec{r}^\pm = \vec{r}_0 \pm \rho\nu(\vec{r}_0)$ for some small $\rho>0$ and $\vec{r}_0\in\Gamma$
    (assuming  that the boundary is of class $C^2$ so the representation of $\vec{r}^\pm$ is 
    unique for $\vec{r}^{\pm}$ near the boundary).
    Clearly
    \begin{eqnarray}\lim_{\rho\to 0}u_S(\vec{r}^+) = \lim_{\rho\to 0} u_S(\vec{r}^-) = u_S(\vec{r}_0) . \end{eqnarray}
 On the other hand, write
    \begin{align}
     u_D(\vec{r}^+) &= \psi_D(\vec{r}_0) \int_\Gamma  \frac{\partial }{\partial \nu'}\Phi_0(\vec{r}^+,\vec{r}') 
      \d S(\vec{r}')  
      + v(\vec{r}^+)  \\
     u_D(\vec{r}^-) &= \psi_D(\vec{r}_0) \int_\Gamma  \frac{\partial }{\partial \nu'}\Phi_0(\vec{r}^-,\vec{r}') 
      \d S(\vec{r}')  
      + v(\vec{r}^-) 
    \end{align}
so that  
    \begin{align}
     v(\vec{r}^\pm) &= \int_\Gamma  \frac{\partial }{\partial \nu'}\Phi (\vec{r}^{\pm},\vec{r}') 
     \big( \psi_D(\vec{r}') - \psi_D(\vec{r}_0) \big) \d S(\vec{r}')
   \nonumber  \\ & \quad + \psi_D(\vec{r}_0) \int_\Gamma 
   \left(\frac{\partial }{\partial \nu'} 
   \Phi (\vec{r}^\pm,\vec{r}') - \frac{\partial }{\partial \nu'}\Phi_0(\vec{r}^\pm,\vec{r}') \right) 
     \d S(\vec{r}') \label{24}
    \end{align}
        where $\Phi_0$ is the Green's function for Laplace equation. It is easy to see that integral \eqref{24} is continuous  in the neighborhood of $\rho=0$. 
        
 The jump condition
   \begin{eqnarray}
  \lim_{\rho\to 0} \left(u_D(\vec{r}^+)  - u_D(\vec{r}^-) \right)= \psi_D(\vec{r}_0) 
  \end{eqnarray}
 now  follows from the calculation 
    \begin{align} \int_\Gamma \frac{\partial }{\partial \nu'}\Phi_0(\vec{r}^\pm,\vec{r}') 
      \d S(\vec{r}')
      &= {1\over 2} \int_{\partial  B_\rho(\vec{r}_0)} \frac{\partial }{\partial \nu'}\Phi_0(\vec{r}^\pm,\vec{r}') 
      \d S(\vec{r}')
      \\
      &= \frac{1}{4\pi p}\int_{\partial  B_\rho(\vec{r}_0)} \pm 1
      \d S(\vec{r}')      
     { \longrightarrow} \pm \frac{1}{2},\quad \rho\to 0    \end{align}
by  applying the divergence theorem, integrating over the circle $B_\rho(\vec{r}_0)$ of radius $\rho$  and 
    shrinking radius $\rho$ to 0.

\end{document}